\documentclass[centertags,12pt]{amsart}

\usepackage{amssymb}
\usepackage{amsmath}
\usepackage{bm}
\usepackage{amsthm}
\usepackage{color}
\usepackage[]{algorithm2e}

\makeatletter
\renewcommand{\subsection}{\@startsection
{subsection}{2}{0mm}{\baselineskip}{-0.25cm}
{\normalfont\normalsize\em}}
\makeatother

\newtheorem{proposition}{Proposition}
\newtheorem{corollary}{Corollary}
\newtheorem{lemma}{Lemma}
{\theoremstyle{definition}

\newtheorem{example}{Example}}
\theoremstyle{remark}
\newtheorem{remark}{Remark}

\begin{document}


\title[Computing sharp recovery structures for LRC codes]{Computing sharp recovery structures for Locally Recoverable codes}

\author{Irene M\'arquez-Corbella} 
\address{Department of Mathematics, Statistics and O. Research, University of La Laguna, 38200 La Laguna , Tenerife, Spain}
\email{imarquec@ull.es}

\author{Edgar Mart\'{\i}nez-Moro} 
\address{Institute of Mathematics,  University of Valladolid, Campus "Miguel Delibes", 47011, Valladolid, Castilla, Spain} 
\email{edgar.martinez@uva.es}

\author{Carlos Munuera} 
\address{Department of Applied Mathematics, University of Valladolid, Avda Salamanca SN, 47014 Valladolid, Castilla, Spain}
\email{cmunuera@arq.uva.es}

\begin{abstract} 
A locally recoverable code is an error-correcting code  such that any erasure in a single coordinate of a codeword can be recovered from a small subset of other coordinates.  In this article we develop an algorithm that computes a recovery structure as concise posible for an arbitrary linear code $\mathcal{C}$ and a recovery method that realizes it. This algorithm also provides the locality and the dual distance of $\mathcal{C}$. Complexity issues are studied as well. Several examples are included.
\end{abstract}

\maketitle


\section{Introduction}
\label{section1}

Locally recoverable  codes  were introduced in \cite{GHSY}, motivated by the  use of coding theory techniques applied to distributed and cloud storage systems, in which the information is  distributed over several nodes.
The growth of the amount of stored data  make the loss of information due to node failures a major problem.
To obtain a reliable storage,  when a node fails we want to be able to recover the data it contains by using information from the other nodes. This is called the {\em repair problem}. A naive method to solve it, is to replicate the same information in different nodes. A more clever strategy is to protect the data by using error-correcting codes, \cite{PD,RKSV}. 
As typical examples of this last solution, we can mention Google and Facebook, that use Reed-Solomon (RS) codes in their storage systems. The procedure is as follows: the information to be stored is a long sequence $b$ of symbols, which are elements of a finite field $\mathbb{F}_{\ell}$. This sequence is cut into blocks, $b=b_1,b_2,\dots$, of the same length, say $t$. According to the isomorphism $\mathbb{F}_{\ell}^t\cong \mathbb{F}_{\ell^t}$, each of these blocks can be seen as an element of the finite field $\mathbb{F}_q$, $q=\ell^t$. Fix an integer $k<q$. The vector $\mathbf{b}=(b_1,\dots,b_k)\in\mathbb{F}_q^k$ is encoded by using a RS code of dimension $k$ over $\mathbb{F}_q$, whose length $n$, $k<n<q$, is equal to the number of nodes that will be used in its storage.  Then we choose $\alpha_1,\dots,\alpha_n\in\mathbb{F}_q$,  and send $b_1+b_2\alpha_i+\dots+b_k\alpha_i^{k-1}$ to the $i$-th node. When a node fails, we may recover the data it stores 
by using Lagrangian interpolation from the information of any other $k$ available nodes.

The above solution to the repair problem is not optimal. When $k$ is small with respect to $n$, then the transmission rate $k/n$ obtained by our encoding method is poor. For large values of $k$ the scheme is wasteful since $k$ symbols must be used to repair just one. Thus it is natural to wonder if there exists other codes allowing the repair of lost encoded  data more efficiently than RS codes, that is by making use of smaller amount of information.

Roughly speaking we can set the repair problem in terms of coding theory as follows:
Let $\mathcal{C}$ be a linear code of length $n$, dimension $k$ and minimum distance $d$ over the field $\mathbb{F}_{q}$.
A coordinate $i\in\{ 1,\dots,n\}$ is {\em locally recoverable with locality $r$} if there is a {\em recovery set} $R\subseteq \{1,\dots,n\}$ with $i\not\in R$ and $\# R= r$, such that for any codeword $\mathbf{x}\in\mathcal{C}$, an erasure in position $i$ of $\mathbf{x}$  can be  recovered by using the information given by the coordinates of $\mathbf{x}$ with indices in $R$.  A collection of recovery sets for all coordinates is a {\em recovery structure}.
The code $\mathcal{C}$ is {\em locally recoverable (LRC) with locality} $\le r$  if  there exists a recovery structure of locality $\le r$, that is to say, if any coordinate is locally recoverable with locality at most $r$. The {\em locality} of $\mathcal{C}$, $\mbox{loc}(\mathcal{C})$ is the smallest $r$ verifying this condition. For example, it is not difficult to prove that MDS codes of dimension $k$ have locality $k$.

Every code $\mathcal{C}$ with minimum distance $d>1$ is locally recoverable with locality $\mbox{loc}(\mathcal{C})\le k$. In practice we are interested in LRC's   admitting recovery sets as small as possible, in relation to the other parameters  $[n, k, d]$ of $\mathcal{C}$. Thus the locality has become a fundamental parameter of a code when it is used for local recovery purposes. Unfortunately the explicit computation of recovery structures, and even the computation of the locality of a specific code, have been revealed as difficult problems. As regards the latter, there exist some known bounds on it. Perhaps the most interesting of them is the following Singleton-like bound:
the locality of $\mathcal{C}$ verifies the relation, \cite{GHSY},
\begin{equation} \label{LRCbound}
k+d+\left\lceil{\frac{k}{\mbox{loc}(\mathcal{C})}}\right\rceil\le  n+2
\end{equation}
which gives a lower bound on $\mbox{loc}(\mathcal{C})$. However it is known that this bound is not sharp, see Example \ref{ej1} below. Codes reaching equality in (\ref{LRCbound}) are called {\em optimal}.

Much research has been devoted in recent years to the repair problem and many recovery structures are known for different types of codes. See for instance \cite{BH,JMX,TB,TBGC} to be aware of the variety of methods used to that purpose. Nevertheless for most of the recovery structures currently available in the literature, it is unknown whether or not they can be refined to obtain other simpler ones. 

In this article we develop an algorithm that computes a recovery structure as concise as posible for an arbitrary code $\mathcal{C}$. This algorithm also gives the locality and the dual distance of $\mathcal{C}$. 
The article is organized as follows: in Section \ref{sectionRecoveryStructures} we summarize all necessary facts about LRC codes and recovery structures that we shall need in the rest. The algorithm is developped in Section \ref{Sharp recovery structure}, where complexity issues are treated as well. Finally, in Section \ref{experimental} we present some experimental results and running times for several examples of codes.

\section{Recovery Structures}
\label{sectionRecoveryStructures}

In this section we give some formal definitions and facts that will be used in the rest of this article. 

Let $\mathcal{C}$ be a linear code of length $n$ over $\mathbb{F}_{q}$. Let $\mathbf{G}$ be a generator matrix of $\mathcal{C}$ and  $\mathbf{c}_1,\dots,\mathbf{c}_n$ be its columns. A set $R\subset\{ 1,\dots,n\}$ with $i\notin R$, is a {\em recovery set} for coordinate $i$ if $\mathbf{c}_i$  is a linear combination of $\{\mathbf{c}_j : j\in R\}$, see \cite{GHSY}. Note that, this is equivalent to say that $\dim \mathcal{C}(R)=\dim \mathcal{C}(R\cup\{i\})$, where $\mathcal{C}(S)$ is the projection of $\mathcal{C}$ on the coordinates in $S$ (see Proposition \ref{propREC} below). Thus the notion of recovery set does not depend on the chosen generator matrix. As we shall prove later, if $R$ is a recovery set for $i$, then for every codeword $\mathbf{x}\in \mathcal{C}$, the coordinate $x_i$ can be obtained from the other coordinates $ x_j$ with indices  $j\in R$.

By an {\em elementary recovery structure} for $\mathcal{C}$ we mean a family $\mathcal{R}=(R_i)_{i=1,\dots,n}$,  such that for all $i=1,\dots,n$,  $R_i$ is a recovery set for the coordinate $i$. Thus a recovery structure allows us to recover an erasure at any position in a codeword.  The structure $\mathcal{R}$ is called {\em minimal} if so is each  $R_i$, that is, if no proper subset of $R_i$ is a recovery set for $i$.   A {\em general recovery structure} for $\mathcal{C}$ is the union of elementary structures, that is to say a collection of recovery sets for each coordinate.  From now on, all structures considered in this article will be elementary.

The code $\mathcal{C}$ is {\em locally recoverable (LRC)} if it admits a recovery structure  $\mathcal{R}=(R_i)_{i=1,\dots,n}$. In such case, the number $\mbox{loc}_i(\mathcal{R})  =\# R_i$  is  the {\em locality} of  $\mathcal{R}$ with respect to coordinate $i$, $1\le i\le n$.  The {\em locality} of  $\mathcal{R}$ is $\mbox{loc}(\mathcal{R})=\max\{ \#R_i : 1\le i\le n \}$. 
As different recovery structures are possible for the same code, it is natural to ask if given one of them, $\mathcal{R}$, there exits another, $\mathcal{R}'$, with smaller recovery sets. So we define
$\mbox{loc}_i(\mathcal{C})=\min\{ \mbox{loc}_i(\mathcal{R}) : \mathcal{R} \mbox{ is a recovery structure for $\mathcal{C}$}\}$, $1\le i\le n$, and  $\mbox{loc}(\mathcal{C})=\min\{ \mbox{loc} (\mathcal{R}) : \mathcal{R} \mbox{ is a recovery structure for $\mathcal{C}$}\}$. If $\mbox{loc}_i(\mathcal{R})=\mbox{loc}_i(\mathcal{C})$ for all $i=1,\dots,n$, then $\mathcal{R}$ will be called {\em sharp}. Clearly  $\mathcal{R}$ is sharp if and only if  $\sum_{i=1}^n \# R_i\le \sum_{i=1}^n \# R'_i$ for any structure $\mathcal{R'}=(R'_i)_{i=1,\dots,n}$ of $\mathcal{C}$, so the locality of a code is always reached through a sharp structure.

\begin{remark}
A structure $\mathcal{R}$ of $\mathcal{C}$ is {\em optimal} if $\mbox{loc} (\mathcal{R})$ reaches equality in the following bound
\begin{equation} \label{LRCboundR}
k+d+\left\lceil{\frac{k}{\mbox{loc}(\mathcal{R})}}\right\rceil\le  n+2
\end{equation} 
which is derived from (\ref{LRCbound}). Note that optimality is not enough to ensure sharpness. For instance, in \cite[Example 1]{TB}  the authors show a $[9,4,5]$ LRC code with a recovery structure $\mathcal{R}$ formed by sets of cardinality 2. Then $\mathcal{R}$ is optimal. Let $\mathcal{R}'$ be the structure whose sets are obtained by adding a (random) coordinate to the recovery sets of $\mathcal{R}$. Thus  $\mbox{loc} (\mathcal{R})=3$ and so $\mathcal{R}$ is optimal as well, but not sharp. 
This example also shows that the locality of a code can not, in general, be obtained from (\ref{LRCbound}), not even when an optimal structure is available.
\end{remark}

Let $d$ be the minimum distance of $\mathcal{C}$. If $d=1$ then, up to reordering, $\mathcal{C}$ contains the codeword $(1,0,\dots,0)$, so the first coordinate  can not have any recovery set and $\mathcal{C}$ is not a LRC. At the other end, if  there exists a coordinate $i$ such that $x_i=0$ for all codeword $\mathbf{x}\in \mathcal{C}$ (that is, if $\mathcal{C}$ is a degenerate code), it is not necessary to recover this coordinate from the others.  So in all that follows we will assume that $\mathcal{C}$ is a nondegenerate code of minimum distance $d>1$. 
Let us investigate in a little more in detail the recovering properties for these codes. As a notation,  $\mathcal{C}^{\perp}$ will be the dual of  $\mathcal{C}$. The {\em support} of a vector $\mathbf{x}\in \mathbb{F}_q^{n}$ is the set $\mbox{\rm supp}(\mathbf{x})=\{ i\in\{1,\dots,n\} : x_i\neq 0 \}$ and its weight is $\mbox{\rm wt}(\mathbf{x})=\#\mbox{\rm supp}(\mathbf{x})$.

\begin{proposition} \label{propREC}
Let $\mathcal{C}$ be a code of length $n$ and let $R\subseteq \{1,\dots,n\}$.  The following statements are equivalent.\newline
(i) $R$ is a recovery set for coordinate $i$; \newline
(ii) $\dim \mathcal{C}(R)=\dim \mathcal{C}(R\cup\{i\})$; \newline
(iii) there exists a codeword $\mathbf{w}\in \mathcal{C}^{\perp}$ such that $i\in \mbox{\rm supp}(\mathbf{w})\subseteq R\cup\{i\}$.
\end{proposition}
\begin{proof}
The equivalence between (i) and (ii) follows from the fact that a generator matrix of  $\mathcal{C}(R)$ can be obtained from the submatrix of $\mathbf{G}$ given by the columns with indices in $R$, and so  $\dim \mathcal{C}(R)=\mbox{rank}  \{ \mathbf{c}_j : j\in R\}$.
Let $\mathbf{w}\in \mathbb{F}_q^{n}$. The equivalence between (i) and (iii) follows from the fact that  $\mathbf{w}\in \mathcal{C}^{\perp}$ if and only if  $\mathbf{G}\mathbf{w}^{T}=\mathbf{0}$, that is if and only if $\sum w_j\mathbf{c}_j=\mathbf{0}$.
\end{proof}

Given a word $\mathbf{w}\in \mathcal{C}^{\perp}$, for any codeword $\mathbf{x}\in \mathcal{C}$  we have that $\mathbf{w}\cdot\mathbf{x}=0$, where $\cdot$ stands for the usual inner product in $\mathbb{F}_q^{n}$, $\mathbf{w}\cdot\mathbf{x}=w_1x_1+\dots+w_nx_n$.  So, if $i\in\mbox{supp}(\mathbf{x})\subseteq \mbox{supp}(\mathbf{w})$, then $x_i$ can be obtained from $\{x_j: j\in \mbox{supp}(\mathbf{w})\setminus\{i\}\}$ as 
\begin{equation} \label{recuperacion}
x_i=-w_i^{-1} (\mathbf{w}\cdot\mathbf{x}).
\end{equation} 
Thus such $\mathbf{w}$ provides a recovery set and a recovery method for coordinate $i$.

\begin{corollary} \label{corREC1}
Let $\mathbf{w}\in \mathcal{C}^{\perp}$,   $\mathbf{w}\neq\mathbf{0}$. Then $\mbox{\rm supp}(\mathbf{w})\setminus\{i\}$ is a recovery set for all $i\in \mbox{\rm supp}(\mathbf{w})$. 
\end{corollary}

A codeword $\mathbf{x}\in \mathcal{C}$ is called {\em minimal} if for every $\mathbf{y}\in \mathcal{C}$, $\mathbf{y}\neq \mathbf{0}$, such that  $\mbox{supp}(\mathbf{y})\subseteq \mbox{supp}(\mathbf{x})$ we have  $\mbox{supp}(\mathbf{y})= \mbox{supp}(\mathbf{x})$. 

%

Given a coordinate $i\in \{1,\dots,n\}$, the codeword $\mathbf{x}$ is said to be {\em $i$-minimal} if $i\in \mbox{supp}(\mathbf{x})$ and  for every $\mathbf{y}\in \mathcal{C}$,  such that  $i\in \mbox{supp}(\mathbf{y})\subseteq \mbox{supp}(\mathbf{x})$ we have  $\mbox{supp}(\mathbf{y})= \mbox{supp}(\mathbf{x})$.

\begin{lemma} \label{iminimal}
Let $\mathbf{x}\in \mathcal{C}$. The following conditions are equivalent. \newline
(i) $\mathbf{x}$ is minimal; \newline
(ii) $\mathbf{x}$ is $i$-minimal for all $i\in\mbox{\rm supp}(\mathbf{x})$; \newline
(iii) $\mathbf{x}$ is $i$-minimal for some $i\in\mbox{\rm supp}(\mathbf{x})$.
\end{lemma}
\begin{proof}
(i)$\Rightarrow$(ii)$\Rightarrow$(iii) are clear. Let us prove (iii)$\Rightarrow$(i). Let  $\mathbf{x}$ be $i$-minimal for some $i\in\mbox{supp}(\mathbf{x})$. If $\mathbf{x}$ were not minimal, it would exist $\mathbf{y}\in \mathcal{C}$,  $\mathbf{y}\neq \mathbf{0}$, such that  $\mbox{supp}(\mathbf{y})\subset \mbox{supp}(\mathbf{x})$. Furthermore $y_i=0$, since  $\mathbf{x}$ is $i$-minimal. Let $j\in\mbox{supp}(\mathbf{y})$ and let $\mathbf{z}=\mathbf{x}-(x_jy_j^{-1})\mathbf{y}\in \mathcal{C}$. Then $z_i=x_i\neq 0$ and $\mbox{supp}(\mathbf{z})\subset\mbox{supp}(\mathbf{x})$, which contradicts that $\mathbf{x}$  is $i$-minimal
\end{proof}

\begin{corollary} \label{corREC2}
Let $\mathcal{C}$ be a code of length $n$ and let $i\in R\subseteq \{1,\dots,n\}$.  The following statements are equivalent.\newline
(i) $R$ is a minimal recovery set for coordinate $i$; \newline
(ii) there exists a minimal codeword $\mathbf{w}\in \mathcal{C}^{\perp}$ such that $\mbox{\rm supp}(\mathbf{w})=R \cup\{i\}$.
\end{corollary}
\begin{proof}
(i)$\Rightarrow$(ii): Let $\mathbf{w}$ be the codeword ensured by item (iii) of Proposition  \ref{propREC}. If  either $\mathbf{w}$ were not $i$-minimal or $\mbox{\rm supp}(\mathbf{w})\neq R\cup\{i\}$, then it would exists  $\mathbf{v}\in \mathcal{C}^{\perp}$ such that $i\in \mbox{\rm supp}(\mathbf{v})\subset R\cup\{i\}$. According to Corollary \ref{corREC1}, $S=\mbox{supp}(\mathbf{v})\setminus\{i\}$ is a recovery set for $i$, which contradicts that $R$ is minimal. Lemma \ref{iminimal} ensures that $\mathbf{v}$ is a minimal word of $\mathcal{C}^{\perp}$. 
(i)$\Rightarrow$(ii): $R=\mbox{supp}(\mathbf{w})\setminus\{i\}$ is a recovery set for $i$, again according to Corollary \ref{corREC1}.
If $R$ were not minimal, then neither would $\mathbf{w}$ be $i$-minimal, according to  Proposition \ref{propREC}.
\end{proof}

So a recovery structure $\mathcal{R}=(R_i)_{i=1,\dots,n}$ for  $\mathcal{C}$ is equivalent to a family $\mathcal{W}=(\mathbf{w}_i)_{i=1,\dots,n}$ of words of $\mathcal{C}^{\perp}$ such that $i\in\mbox{supp}(\mathbf{w}_i)$: just take $R_i=\mbox{supp}(\mathbf{w}_i)\setminus \{i\}$. Besides, as seen before, the equalities $\mathbf{w}_i\cdot\mathbf{x}=0$, $i=1,\dots,n$, provide a method to compute any erased coordinate $x_i$ in a word $\mathbf{x}\in\mathcal{C}$. $\mathcal{R}$ is minimal iff so is each $\mathbf{w}_i$; and $\mathcal{R}$ is sharp iff $\sum_{i=1}^n \# \mbox{supp}(\mathbf{w}_i)$ is minimal among all families $\mathcal{W}'=(\mathbf{w}'_i)_{i=1,\dots,n}$ of words of $\mathcal{C}^{\perp}$ such that $i\in\mbox{supp}(\mathbf{w}'_i)$.

\begin{corollary}
Let $\mathcal{C}$ be a code and let $\mathcal{C}^{\perp}$ be its dual. Then ${\rm loc}(\mathcal{C})\ge d(\mathcal{C}^{\perp})-1$.
\end{corollary}

The bound on ${\rm loc}(\mathcal{C})$ given by the previous corollary  improves sometimes the bound given by (\ref{LRCbound}), which shows again that the latter is not necessarily sharp.

\begin{example}\label{ej1}
Let us consider a $[9,4]$ code $\mathcal{C}$ over $\mathbb{F}_4=\{ 0, 1, \alpha, 1+\alpha\}$ with $\alpha^2=1+\alpha$. According to the Griesmer bound, its minimum distance verifies that $d(\mathcal{C})\le 5$. Indeed such a code with $d(\mathcal{C})=5$ exists. Consider, for example, the one given by the generator matrix
$$
G=\left(
\begin{array}{ccccccccc}
 1&0&0&0&\alpha& 1+\alpha& \alpha& 1+\alpha&1 \\  
 0&1&0&0& 1+\alpha& \alpha&1& 1+\alpha&0 \\
 0&0&1&0&0&\alpha&1+\alpha&1&1+\alpha \\
 0&0&0&1&\alpha&1 & \alpha&0&1+\alpha 
\end{array} \right)\in \mathbb F_4^{4\times 9}. 
$$
$\mathcal{C}$ is obtained from the $[12,6]$ extended Quadratic Residue code by shortenning twice and then puncturing; its minimum distance is $5$, see \cite{mint}.  From (\ref{LRCbound}) its locality verifies ${\rm loc}(\mathcal{C})\ge 2$. But being $\mathcal{C}$ an almost MDS code, its dual $\mathcal{C}^{\perp}$ must be a $[9,5,4]$ code, see \cite[Theorem 7]{FW}. Thus ${\rm loc}(\mathcal{C})\ge d(\mathcal{C}^{\perp})-1= 3$.
In fact, as we shall compute later in Example \ref{excont}, we have ${\rm loc}(\mathcal{C})= 3$.
\end{example}

\begin{remark}
There exists a remarkable connection between the repair problem and the theory of secret sharing.  A {\em secret sharing scheme (SSS)} is a  method for distributing a secret $s\in\mathbb{F}_{q}$ among a set of $n$ participants, 
each of whom receives a piece of the secret or {\em share}. The scheme is designed in such a way that only the authorized 
coalitions of participants can recover the secret, by pooling the shares of its members. 

SSS can be obtained from codes in several ways. The first method to do that was given in \cite{SSS}. 
Let  $\mathcal{C}$ be a $[n+1,k]$ code with generator matrix $\mathbf{G}$. A dealer computes the codeword $\mathbf{x}=\mathbf{z}\mathbf{G}$ from a random vector $\mathbf{z} \in\mathbb{F}_{q}^k$ subject to $x_{n+1}=s$, the secret to be shared. The share of the $i$-th participant is $x_i$, $1\le i\le n$.  Thus, a coalition $R$ of participants is authorized if and only if $R$ is a recovery set for coordinate $n+1$. The locality $\mbox{\rm loc}_{n+1}(\mathcal{C})$ is exactly the smaller size of an authorized coalition. Massey \cite{MAS} later introduced the concept of  minimal codeword and observed that the minimal authorized coalitions correspond to the minimal words of  $\mathcal{C}^{\perp}$ whose support contains the coordinate $n + 1$, that is,  to the minimal recovery sets for the coordinate $n + 1$. Therefore, the methods developed in this article can be adapted to compute the minimal authorized coalitions in this type of SSS.
\end{remark}

\section{Computing Sharp Recovery Structures}
\label{Sharp recovery structure}

In this section we develop an algorithm that provides a sharp recovery structure for an arbitrary linear code  $\mathcal{C}$, in terms of minimal codewords of its dual  $\mathcal{C}^{\perp}$, that is through  a family $\mathcal{W}=(\mathbf{w}_i)_{i=1,\dots,n}$ of minimal words of $\mathcal{C}^{\perp}$ such that $i\in\mbox{supp}(\mathbf{w}_i)$. Consequently, it also provides a method for recovery, following the formula of equation (\ref{recuperacion}), the locality of  $\mathcal{C}$ and its dual distance, as the smallest weight of a word in   $\mathcal{W}$.
The methods we will use come from the theory of  Gr\"obner basis, however we will try to avoid  Gr\"obner basis terminology.

Let $\zeta $ be a primitive element of $\mathbb F_q^*$. We define a total ordering $\prec_T$ on $\mathbb F_q^n$ as follows
\begin{equation}
 \mathbf x \prec_T \mathbf y \hbox{ if } \begin{cases}
\mbox{wt}(\mathbf x)< \mbox{wt}(\mathbf y); \hbox{ or}
\\
\mbox{wt}(\mathbf x)= \mbox{wt}(\mathbf y) \hbox{ and } (i_{x_1},\ldots,i_{x_n})  {\prec_{lex}} (i_{y_1},\ldots,i_{y_n})
\end{cases}
\end{equation}
where $i_{x_j}$ is the  element of  $[0,q-1]\cup \{-\infty\}$ such that $\zeta^{i_{x_j}}= x_j$ or $i_{x_j}=-\infty $ if $x_i=0$, and the same definition holds for the vector $\mathbf y$. Note that our ordering $\prec_T$ is mainly the order given by the Hamming weight with an extra total ordering for breaking ties. Remark that any other total order in the indices, different from $lex$ for breaking ties, will define an alternative order $\prec'_T $ that would be equally valid for our purposes.

\begin{example}
Let us consider the vectors $\mathbf x = (1,\alpha, 0, \alpha^2)$ and $\mathbf y = (\alpha^2, \alpha, 0 , \alpha)$ over $\mathbb F_4^4$, being $\alpha$ a primitive element of $\mathbb F_4^*$ with $\alpha^2 = 1+\alpha$. Then, $\mathbf x  \prec_T \mathbf y$ since $\mbox{wt}(\mathbf x) = \mbox{wt}(\mathbf y)$ and $(0,1,-\infty, 2)  {\prec_{lex}}  (2,1,0,1)$.
\end{example}

Next we recall some concepts from Gr\"obner basis theory that we will use in what follows.
We say that a codeword $\mathbf s\in \mathcal C$ is expresed as a {\em syzygy} if we write $\mathbf s=\mathbf s_1-\mathbf s_2$,  with $\mathbf s_1,\mathbf s_2\in \mathbb F_q^n$ and $\mathbf s_2 \prec_T\mathbf s_1$. In this case  $\mathbf s_1$ is called the {\em leading term} of the syzygy. Every codeword can be trivially  expresed in this form, by taking $\mathbf s_2=\mathbf 0$.  Given two vectors $\mathbf{t},\mathbf{x}\in \mathbb F_q^n$, we say that  
$\mathbf t$ {\em reduces}  $\mathbf{x}$  if $\mathbf{x}+ \mathbf{t} \prec_T \mathbf{x} $.
Then, the vector  $\mathbf{v}=\mathbf{x}+\mathbf{t}$ is called the  {\em reduction} of $\mathbf{x}$ by $\mathbf{t}$ and will be denoted by $\mathbf{x}\rightarrow_{\mathbf{t}}\mathbf{v}$. Let $\mathcal{T}\subseteq \mathbb{F}_q^n$. The vector $\mathbf{v}$ is a reduction of  $\mathbf{x}$  {\em by} $\mathcal{T}$ if  for a non-negative integer $s\in\mathbb N$ there are elements $\lambda_i\in \mathbb F_q$ and $\mathbf{t}_i\in \mathcal{T}$ $i=1,\ldots ,s$ such that
\begin{equation}\label{eq:red} \mathbf{x}\rightarrow_{\lambda_1\cdot\mathbf{t}_1}\mathbf{v}_1,\, \mathbf{v}_1\rightarrow_{\lambda_2\cdot\mathbf{t}_2}\mathbf{v}_2,\, \ldots \, \mathbf{v}_{n_1}\rightarrow_{\lambda_s\cdot\mathbf{t}_s}\mathbf{v},\end{equation}
and we will denote it by $\mathbf{x}\rightarrow_{\mathcal{T}}\mathbf{v}$. Note that, in general, a vector has not a unique reduction by a set  $\mathcal{T}$.

\begin{example}
Given the set 
$\mathcal T = \left\{ \mathbf t_1 = (0,1,1,0,0,1), \mathbf t_2 = (1,1,0,1,1,0)\right\}\subseteq \mathbb F_2^6$ and 
the vector $\mathbf x = (1,1,1,1,1,0)\in \mathbb F_2^6$. We can give two different reductions of $\mathbf x$ by $\mathcal T$:
$$\begin{array}{ccc}
\mathbf{x}\rightarrow_{\mathbf{t}_1}\mathbf{v}_1 = (1,0,0,1,1,1)
 & \hbox{ and } &
\mathbf{x}\rightarrow_{\mathbf{t}_2}\mathbf{v}_2 = (0,0,1,0,0,0)
\end{array}$$
\end{example}


We will compute a special set $\mathcal{T}\subseteq \mathbb{F}_q^n$ of syzygies called a Gr\"obner test-set as follows.
\begin{enumerate}
\item First we list $\tt L$ all the non-zero syzygies of the code $\mathcal C$ and we sort them with respect to the ordering $\prec_T$ in its leading terms.
\item Then, we go incrementally through the list $\tt L$. We add to the set $\mathcal T$ and we remove from $\tt L$ all those elements such that at least one of its non-zero coordinates from its leading term is equal to a multiple of the leading term of an element in $\mathcal T$.
\item Finally we reduce the trailing terms using the list $\tt L$.
\end{enumerate}

\begin{example}\textbf{[Toy Example]}
Consider a binary code $\mathcal C$ with generator matrix 
$$G=\left(\begin{array}{ccc} 1 & 0 & 1 \\ 0 & 1 & 1 \end{array}\right)\in \mathbb F_2^{2\times 3}.$$
To obtain a Gr\"obner test-set $\mathcal T$ for $\mathcal C$ we proceed as follows.

We initialize $\tt L$ with all the non-zero syzygies of the code $\mathcal C$ and we sort them w.r.t. $\prec_T$.
$${\tt L} = \left\{ \begin{array}{cc}
(0,1,0) - (0,0,1), &
(1,0,0)-(0,0,1), \\
(1,0,0)-(0,1,0), &
(0,1,1)-(0,0,0),\\ 
(1,0,1)-(0,0,0), &
(1,1,0)-(0,0,0)
\end{array}
\right\}$$
\begin{itemize}
\item We add $(0,1,0)-(0,0,1)$ to $\mathcal T$ and we omit from $\tt L$ all multiples of $(0,1,0)$.
\item We add $(1,0,0)-(0,1,0)$ to $\mathcal T$ and we omit from $\tt L$ all multiples of $(1,0,0)$.
\end{itemize}
The process ends because the list $\tt L$ is empty. Now $\mathcal T$ is a Gr\"obner test-set.
\end{example}

Note that $\mathcal{T}$ depends only on the code $\mathcal{C}$ and the ordering $\prec_T$ and  can be computed from a set of generators  by means of the  Gr\"obner basis procedure based on linear algebra FGLM techniques, stated in \cite{MMS}.  
 Moreover the reduction of  a vector in $\mathbb F_q^n$ by this $\mathcal T$ is show to be unique and independent of the order of the reductions in (\ref{eq:red}).
 Algorithms \ref{Algor-1} gives the reduction of an element $\mathbf{x}$ by a Gr\"obner test set $\mathcal{T}$.

\begin{algorithm}[H] \label{Algor-1}
\KwData{A  Gr\"obner test set $\mathcal{T}\subseteq \mathbb{F}_q^n$ and a vector $\mathbf x\in \mathbb F_q^n$.}
\KwResult{Reduced form of $\mathbf x$}
initialization  $\mathbf r \leftarrow \mathbf 0$\;
\While{there exist  $\mathbf t\in \mathcal T$ and  $\lambda\in\mathbb F_q$ such that $\mathbf x-\lambda \mathbf t \prec_T \mathbf x$\;}{
 $\mathbf r \leftarrow \mathbf r -\lambda \mathbf t$ \;}
\caption{Reduction by the Gr\"obner test set $\mathcal T$.}
\end{algorithm}

The Gr\"obner  test set $\mathcal{T}$ defines a reduction that allows to check whether a  vector  is in the code $\mathcal{C}$ if and only it reduces to the zero vector. Furthermore it allows us to compute a sharp recovery structure for $\mathcal C$.

\begin{proposition}
Let $\mathcal T$ be a Gr\"obner test set for the code $\mathcal C$ and let $\mathbf{x}\in\mathbb{F}_q^n$. Then $\mathbf{x} \in \mathcal{C}$ if and only if $\mathbf{x}$ reduces to $\mathbf 0$ by $\mathcal{T}$. 
\end{proposition}
\begin{proof}
For a proof see \cite{MMS}.
\end{proof}

The key idea on which our algorithm is based in the following result.

\begin{proposition} \label{propSRS} 
Let $\mathcal T$ be a Gr\"obner test set for the code $\mathcal C^\perp$. For each coordinate $i$ let $\mathcal{T}_i=\{ \mathbf{t} \in \mathcal{T} : i\in\mathrm{supp}(\mathbf{t})\}$ and let $\mathbf t_i$ be an element of minimum Hamming weight  in $\mathcal{T}_i$.   Then $(\mathrm{supp}(\mathbf{t}_i)\setminus \{i\})_{i=1,\dots,n}$ is a sharp recovery structure for $\mathcal C$.
\end{proposition} 
\begin{proof}  
Let  $\mathbf m_i\in \mathcal C^\perp$ be the smallest $i$-minimal codeword with respect to the total ordering $\prec_T$. We will prove that  $\mbox{wt}(\mathbf{t}_i) = \mbox{wt}(\mathbf{m}_i)$. {Note that since $\mathbf m_i$ is in the code $\mathcal C^\perp$ then  $\mathbf{m}_i\rightarrow_{\mathcal{T}}\mathbf{0}$. There is a non-negative integer $s\in\mathbb N$ such that the set of reductions in Algorithm~\ref{Algor-1} for $\mathbf{m}_i\rightarrow_{\mathcal{T}}\mathbf{0}$ can be  expresed in the form 
\begin{equation}\label{eq:red0} \mathbf{m}_i\rightarrow_{\lambda_1\cdot\mathbf{t}_1}\mathbf{v}_1,\, \mathbf{v}_1\rightarrow_{\lambda_2\cdot\mathbf{t}_2}\mathbf{v}_2,\, \ldots \, \mathbf{v}_{n_1}\rightarrow_{\lambda_s\cdot\mathbf{t}_s}\mathbf{0}.\end{equation}
	
	Thus there must be a vector  $\mathbf t=\lambda_j \mathbf t_j$, $j\in \{1,\ldots ,s\}$   such that it is the first one involved in the chain of reductions in (\ref{eq:red0}) with $i\notin\mathrm{supp}(\mathbf{v_i})$. Indeed it is clear that $j=1$ since $i$ is in the support of the codewords $\mathbf{v}_{1},\ldots , \mathbf{v}_{i-1}$ and $\mathbf{m}_{i} \succ_T \mathbf{v}_{1}\succ_T\ldots \succ_T \mathbf{v}_{j-1}$ which contradices  the minimality of $\mathbf m_i$  with respect to the total ordering $\prec_T$.	Without loss of generality we can suppose that $\mathbf t_j=\mathbf t_i$.
	
	 Let $\mathbf r=  \mathbf t - \mathbf m_i$ and note that $\mathbf r\prec_T \mathbf m_i$ since $\mathrm{wt}(\mathbf r)=\mathrm{wt}(\mathbf v_1)< \mathrm{wt}(\mathbf m_i)$.  Suppose now that $\mbox{wt}(\mathbf m_i) < \mbox{wt}(\mathbf t)= \mbox{wt}(\mathbf t_i)$,   hence  $\mathbf r$  is expresed as the  syzygy $ \mathbf t - \mathbf m_i$. The syzygy $\mathbf r$  can be used to reduce  $\mathbf t$  as follows $ \mathbf{t}\rightarrow_{\mathbf r}\mathbf{m}_i$, since $(\mathbf t-\mathbf r)= \mathbf m_i$ without cancelling the $i$-th position  which is a contradiction with the fact that $\mathbf t_j\in\mathcal T$ since the test set contains syzygies that cannot be reduced. Hence  $\mbox{wt}(\mathbf m_i) = \mbox{wt}(\mathbf t)$ and we are done.}
\end{proof}

According to Proposition \ref{propSRS}, the minimal words $\mathbf{w} _i\in \mathcal{C}^{\perp}$  providing a sharp recovery structure for $\mathcal{C}$ can always be found in a test set $\mathcal{T}$, and it is not necessary to look for them in the  whole code $\mathcal{C}^{\perp}$. This idea is used in Algorithm~\ref{Algor-2} that  provides such a sharp recovery structure for $\mathcal C$. Roughly speaking, this algorithm applies Gaussian elimination to a large sparse matrix (namely list $\mathcal G_T$ in this case),  taking into account that we stop once we have enough codewords to cover all indices in $\{ 1,\ldots ,n\}$. Note that it  only requires a parity check matrix $H$ of the code $\mathcal C$ to run.


\begin{algorithm}[t] \label{Algor-2}

\SetAlgoLined
\SetSideCommentRight
\SetNoFillComment

\KwData{A parity check matrix $H$ of an $[n,k]_q$ linear code $\mathcal C$}
\KwResult{A sharp recovery structure for $\mathcal C$}

\textbf{Initialization}

$\tt List1 \leftarrow$ every row of $H$ and all its multiples in $\mathbb F_q$\;
$\tt List2 \leftarrow $ the set of all $n$-tuples in $\mathbb F_q$ of weight less than $\tt n-k+2$ \;
$\tt Recovery\_struc$\;

$\mathbf v \leftarrow \mathbf 0$; ~~ 
$r\leftarrow 0$\;

\While{the support of $\tt Recover\_struc \neq \{1, \ldots, n\}$}{
	\For{$\tt \mathbf g$ in $\tt List1$}{
	$\tt \mathbf w =(\mathbf w[1], \mathbf w[2]) \leftarrow \left( \mathbf v, \mathbf v + \mathbf g \right)$\;
	\If{$\tt \mathbf w[1]$ is not a multiple of the leading terms of $\tt \mathcal G_T$}{
		\tcc{Let $\tt g=(g[1],g[2])\in \mathcal G_T$, the leading term of $g$ is $\tt g[1]$}
		$\tt j = Member(\mathbf w[2], [\mathbf v_1[2], \ldots, \mathbf v_r[2]])$\;
		\tcc{\small $\tt Member(v,[v_1, \ldots, v_r])$ returns $\tt j$ if $\tt v=v_j$ or $\tt false$ o/w}
		\If{$\tt j \neq false$}{
			$\mathbf g = (g[1], g[2]) \leftarrow \left(\mathbf w[1] , \mathbf v_j[1]\right)$\;
			$\mathcal G_T \leftarrow \mathcal G_T \cup \{ \mathbf g\}$\;
			\If{ the support of $\tt \mathbf g$ is not contained in the support of $\tt Recovery\_struc$}{
				$\tt Recovery\_Struc \rightarrow Recovery\_Struc \cup \{\mathbf g\}$\;
			}
		}
		\Else{
			$\tt r\leftarrow r+1$ ~~ and ~~ 
			$\tt \mathbf v_r \leftarrow \mathbf w$
		}
	
	}
		
	}
	$\tt \mathbf v = NextTerm(List2)$\;
	\tcc{\small $\tt NextTerm(List)$ removes the first element from the list $\tt List$ and returns it.}
}

\caption{Algorithm for computing a sharp recovery structure for $\mathcal C$}
\end{algorithm}

\begin{remark}
If $q=2$, then the parity check matrix of some codes can be seen as the adjacency matrix of a non-directed graph. In these cases, Proposition~\ref{propSRS} ensures that  the procedure in Algorithm \ref{Algor-2}  performs  the Horton's algorithm for computing a minumum cycle basis of the  graph \cite{BBFM,Horton}. Thus, our Algorithm~\ref{Algor-2} extends the Horton's algorithm from cycles in graphs to recovery sets in codes.
\end{remark}

\begin{example} \label{excont}
Let us consider again the code $\mathcal{C}$ of Example~\ref{ej1}.  A parity check matrix of $\mathcal{C}$ is easily obtained from its generator matrix. 
The Gr\"obner test set consists of 49 codewords. The following four among them define a  sharp recovery structure $\mathcal R = \left\{ R_1, \ldots, R_9 \right\}$ for $\mathcal{C}$
$$
\begin{array}{lcl}
\mathbf t_1 = (\alpha, \alpha, \alpha+1, 0,0,0,0,\alpha + 1 , 0) & & \mathrm{supp}(\mathbf t_1) = \left\{ 1,2,3,8\right\}\\
\mathbf t_2 = (\alpha, \alpha+1, 0,\alpha,1,0,0,0, 0) & & \mathrm{supp}(\mathbf t_2) = \left\{ 1,2,4,5\right\}\\
\mathbf t_3 = (\alpha, 1, 0,0,0,1, \alpha+1, 0, 0) & & \mathrm{supp}(\mathbf t_3) = \left\{ 1,2,6,7\right\}\\
\mathbf t_4 = (\alpha, 0, 0,1,1,0,0,0, \alpha) & & \mathrm{supp}(\mathbf t_4) = \left\{ 1,3,4,9\right\}\\
\end{array}
$$
with
$$ 
\begin{array}{ccc} 
R_1 = \mathrm{supp}(\mathbf t_1) \setminus \{1\} &
R_2 = \mathrm{supp}(\mathbf t_1) \setminus \{2\} &
R_3 = \mathrm{supp}(\mathbf t_1) \setminus \{3\}\\
R_4 = \mathrm{supp}(\mathbf t_2) \setminus \{4\} &
R_5 = \mathrm{supp}(\mathbf t_2) \setminus \{5\}&
R_6 = \mathrm{supp}(\mathbf t_3) \setminus \{6\}\\
R_7 = \mathrm{supp}(\mathbf t_3) \setminus \{7\}&
R_8 = \mathrm{supp}(\mathbf t_1) \setminus \{8\}&
R_9 = \mathrm{supp}(\mathbf t_4) \setminus \{9\}.\\
\end{array}  
$$
In particular,  ${\rm loc}(\mathcal{C})=3$ and  $d(\mathcal{C}^{\perp})= 4$, as announced in Example \ref{ej1}.
\end{example}

\begin{proposition}
Algorithm \ref{Algor-2} is correct and provides a sharp recovery structure for the code $\mathcal C$ with parity check matrix $H$.
\end{proposition}
\begin{proof}
The first stage involves the initalization of the algorithm with a list  $\tt List1$ whose elements are the rows of the parity check matrix $H$ and all their  non-zero scalar multiples. That is:
$\tt List1=\left\{ \lambda h_i \hbox{ with } \lambda \in \mathbb F_q\setminus\{0\} \hbox{ and } h_i \hbox{ row of }H\right\}$
with  $(n-k)(q-1)$ items. Now, the list $\mathcal G_T$ that we will use during the algorithm is initialized with the elements of $\tt List1$ and we add to all this elements the zero vector as label (to represent the trivial syzygies of the elements in the list).
In another list, namely $\tt List2$, we have all vectors of $\mathbb F_q^n$ of weight less or equal to $n-k+2$, we sort this list w.r.t. $\prec_T$.

Then, at each step we remove the first element $\mathbf w$ from the list $\tt List2$. 
If any of the vectors $\mathbf w+\mathbf h$ with $\mathbf h$ being an item of $\tt List1$ coincide with the $j$-th element of $\mathcal G_T$ then, the difference between $\mathbf w$ and the label corresponding to $\mathcal G_T[j]$ form a  codeword of minimal support of the dual code of $\mathcal C$. Otherwise, we add the items:
$\left\{ \mathbf w + \mathbf h  \mid \forall \mathbf h \hbox{ item of  } \tt List 1\right\}$
as new elements of the list $\mathcal G_T$ with the vector $\mathbf w$ as label (i.e. we add the new reduced syzygies). We repeat this process until we get enough codewords to achieve a sharp recovery structure for the code $\mathcal C$.
\end{proof}

\begin{proposition}
Let $\mathcal C$ be a linear code of length $n$ and dimension $k$  over $\mathbb F_q$. If the ground field operations need one unit time, then Algorithm \ref{Algor-2} applied to $\mathcal{C}$ takes time $\mathcal O(Dn(n-k) (q-1) \mathrm{log}(q))$, where $D$ is the total number of iterations of  the algorithm. 
\end{proposition}

\begin{proof}
The proof is similar to the proof of \cite[Theorem 4.3]{MMS}, adapted to the changes made in Algorithm \ref{Algor-2} versus \cite[Algorithm 2]{MMS}. The hardest part of our algorithm is the management of the list $\mathcal G_T$. In each main loop iteration, up to $(n-k)(q-1)$ new elements are added to the list $\mathcal G_T$, then compared and finally redundancy is eliminated. Note that comparing two vectors in $\mathbb F_q^n$ requires $\mathcal O(n\log (q))$ field operations.
At iteration $i$, after inserting the new elements in the list $\mathcal G_T$, we have at most
$(q-1) (n-k) + (i (q-1)(n-k) - i)$ elements. Here, the first summand corresponds to the elements that initialized $\tt List1$, while the second one comes from the fact that at each iteration the first element is removed and we add $(n-k)(q-1)$ new elements.
If $D$ is an upper bound for the number of iterations of Algorithm \ref{Algor-2}, this gives a total time of order
$$
\mathcal O\left(n \log(q) \left( (q-1)(n-k) + D(q-1)(n-k)-D \right)\right) \sim \mathcal O\left( D n (n-k) (q-1) \log(q)\right).
$$
\end{proof}

\begin{remark} (1) The number of iterations in Algorithm \ref{Algor-2}, $D$, is upper bounded by the fact that the weight of a minimal codeword is at most $n-k+1$ \cite[Lemma 2.1]{Minimal}. 
Thus, it follows that 
$$D\leq \sum_{i=0}^{n-k+1}\binom{n}{i} (q-1)^i.$$  
(2)  Note that this algorithm also provides the dual distance of $\mathcal{C}^{\perp}$, as the smallest weight of one of the minimal words in the obtained   recovery strucuture of $\mathcal{C}$. Remember that computing the minimum distance of a linear code is a NP complete problem, \cite{NP}, which explains the high complexity of our algorithm. Note, however, that as in the case of the minimum distance, a recovery structure must be calculated only once per code.
\end{remark}

\section{Experimental Results}
\label{experimental}

Algorithm \ref{Algor-2} has been implemented with the program \textsc{Sagemath} \cite{sage}.
In the following tables we summarize the average running times for several examples of codes, obtained with an Intel $\circledR$ CoreTM 2 Duo $2.8$ GHz. The experiments are performed as follows: 
We first generate a full rank random matrix of size $k\times n$ over $\mathbb F_q$  using the command {\tt random\_matrix (GF(q),k,n)}; then we take the corresponding code $\tt C$ and compute its dual $\tt Cd=C.dual\_code()$; if the minimum distances of $\tt C$ and $\tt Cd$ are greater than 1, we apply Algorithm \ref{Algor-2}.

For each base field size $q$, the experiment has been performed on $20$ random codes $\tt C$. The obtained results are shown in Tables \ref{Table-1} and \ref{Table-2}.
In Table \ref{Table-1} all codes have length $n=10$ and dimension $k=4$. The first column contains the base field sized $q$. Second column indicates the average running time for the computation of a sharp recovery structure for $\tt C$, measured in seconds. Third column shows the average number of vectors in $\tt List2$. In Table \ref{Table-2}, we deal with codes of different parameters, which are indicated in the second column. Here we have omitted  the average number of vectors in $\tt List2$.
\begin{table}[h!]
\begin{tabular}{ccc}
\hline
Field size $q$ &  Running time  & Vectors in $\tt List2$ \\ \hline
$q=2$ & $0.009$ seconds & $34$ \\
$q=3$ & $0.022$ seconds & $113$  \\
$q=5$ & $0.322$ seconds & $665$ \\
$q=7$ & $0.727$ seconds & $1010$  \\
$q=11$ & $3.649$ seconds & $4308$\\
$q=13$ & $5.887$ seconds & $5743$\\
$q=17$ & $13.141$ seconds & $10950$\\
$q=19$ & $18.210$ seconds & $14013$\\
$q=23$ & $40.876$ seconds & $22069$ \\
\hline
\end{tabular}
\caption{\small Running times of Algorithm \ref{Algor-2} for random $[10,4]$ codes.}
\label{Table-1}
\end{table}
\begin{table}[h!]
\begin{tabular}{ccc}
\hline
Field size $q$ &  Parameters  & Running time \\ \hline
$q=2$ & $[50, 10]$ &  $0.84$ s. \\
$q=2$ & $[50, 12]$ &  $2.984$ s. \\
$q=2$ & $[50, 15]$ &  $11.103$ s. \\
$q=2$ & $[50, 20]$ &  $215.47$ s. \\
$q=2$ & $[70, 15]$ &  $78.4$ s. \\
$q=3$ & $[50, 10]$ &  $20380.70$ s. \\
$q=5$ & $[25, 7]$ &  $14793.754$ s.\\
\hline
\end{tabular}
\caption{\small Running times of Algorithm \ref{Algor-2} for random codes of different parameters.}
\label{Table-2}
\end{table}


\end{document}